\documentclass[11pt,letterpaper]{llncs}

\newif\ifshowproofs
\showproofsfalse 

\newif\ifnotshowproofs
\ifshowproofs
  \notshowproofsfalse
\else
  \notshowproofstrue
\fi

\usepackage[margin=1in]{geometry}

\usepackage[T1]{fontenc}
\usepackage[utf8]{inputenc}
\usepackage{lmodern} 
\usepackage{natbib}
\usepackage{setspace}
\usepackage{parskip}

\usepackage[utf8]{inputenc}
\usepackage{mathtools}
\usepackage{enumerate}
\usepackage{epstopdf}
\usepackage{amsfonts,amssymb}
\usepackage{epsfig}
\usepackage{makeidx}
\usepackage{algorithm}
\usepackage{algorithmic}
\usepackage{anysize}
\usepackage{bm}
\usepackage{enumerate}
\usepackage{graphicx,color,epsfig,epstopdf}

\newcommand{\setR}{\mathbb{R}}

\newcommand{\demand}{d}

\DeclareMathOperator*{\argmin}{argmin}   
\DeclareMathOperator*{\argmax}{argmax}   

\newcommand{\F}{\mathcal{F}}
{\bfseries}{}

\begin{document}


\title{On the Nucleolus of a Class of Linear Production Games}


\author{Mourad Ba\"iou\and
Gianpaolo Oriolo \and
Gautier Stauffer}

\institute{LIMOS Lab. CNRS/UCA, Clermont-Ferrand, France \email{mourad.baiou@isima.fr}\and
Dipartimento di Ingegneria Civile ed Ingegneria Informatica, Universit\`a Tor Vergata, Rome,  Italy, \email{oriolo@disp.uniroma2.it}  
 \and
Faculty of Business and Economics (HEC Lausanne), Department of Operations, University of Lausanne, Lausanne, Switzerland \email{gautier.stauffer@unil.ch}}


\institute{}

\maketitle


\begin{abstract}
We study the nucleolus in a class of cooperative games that model settings in which multiple agents collaborate by sharing demands and production-distribution capacities for a specific commodity, potentially across multiple markets, to improve overall efficiency. These \emph{production-distribution games} arise in practical applications such as horizontal collaboration in last-mile logistics, and — as we show in this paper — they form a well-structured subclass of the \emph{linear production games} introduced by~\citet{owen1975core}. Linear production games is a subclass of \emph{totally balanced games}, first introduced by~\citet{shapley1969market}, a class of games in cooperative game theory that always admits a nonempty core.

Building upon a characterization for the core of linear production games given by~\citet{owen1975core}, we show that a point in the core may be efficiently found for production-distribution games, and we turn our attention to the computation of the nucleolus. 
Computing the nucleolus is in general NP-hard. As for totally balanced games, positive results are limited to a few notable families, such as minimum-cost spanning tree games~\citep{granot1984core}, assignment games~\citep{solymosi1994algorithm}, and convex games~\citep{faigle2001computation}. While linear production games may appear structurally more tractable,~\citet{deng2009finding} demonstrated that computing the nucleolus in this class remains NP-hard. This underscores the importance of identifying subclasses within these games where the nucleolus can be computed efficiently. Production-distribution games represent a promising subclass in which efficient computation of the nucleolus may still be achievable.

In addition to establishing that production-distribution games are linear production games, our main contributions center on the \emph{uncapacitated} variant of these games. We identify structural properties that enable polynomial-time computation of the nucleolus and introduce algorithmic techniques that differ significantly from those used in previously known tractable cases:

\begin{itemize}
    \item \emph{Core as a singleton.} We give a polynomial-time characterization of instances where the core reduces to a singleton, allowing the nucleolus to be computed directly via linear programming duality. This characterization builds upon conditions that are reasonable in practical settings and interesting from a theoretical point of view as they go well beyond classical sufficient conditions for single-point cores.

\item \emph{Fixed number of markets.} We design a polynomial-time algorithm for each iteration of the Maschler scheme, based on a separation oracle that reduces to polynomially many applications of a general-purpose lemma. Given $a \in \mathbb{R}^p$ and $b \in \mathbb{R}$, the lemma decides in polynomial time whether there exists a subset $S \subseteq \{1,\dots,p\}$ such that $a(S) < b$ and $S \cup Q$ is not contained in the span of prescribed vectors $S_1,\dots,S_k$, for some $Q$. We believe this lemma is of independent interest.

    \item \emph{Single market case.} We develop a faster combinatorial algorithm that exploits the structure of optimal solutions in each iteration of the Maschler scheme. We show that a carefully designed primal-dual algorithm converges in a single step per iteration, yielding an efficient iterative procedure with runtime $O(n^4)$, without relying on low-level optimization. We also show that, in this case, there always exists a family of only $2n-1$ sets that is a characterization set for the nucleolus. A characterization set~\citep{Granot98,reijnierse98} is a collection of coalitions that determines the nucleolus by itself. To the best of our knowledge, the latter result cannot be derived from (sufficient) conditions that are known for characterization sets~\citep{skizlai}.
\end{itemize}

These results provide expand the catalogue of linear production and totally balanced games for which the nucleolus can be computed efficiently and provide new algorithmic tools that may inspire further progress in cooperative game theory. They also raise open questions about the complexity for production-distribution games of nucleolus computation in both the general uncapacitated case, that is when the number of markets is not fixed, and the capacitated case, already with a single market.

\ifnotshowproofs 
IN THIS VERSION OF THE MANUSCRIPT ALL PROOFS ARE IN APPENDIX
\fi
\end{abstract}

\newpage
\section{Introduction}\label{definition1}

We consider a situation where a set $N=\{1,...,n\}$ of $n\in \mathbb{N}\setminus \{0\}$ price-taker companies, producing the same 
commodity product are willing to cooperate to increase their profit. In the most general setting, the commodity product is sold in a set $M=\{1,...,m\}$ of $m \in \mathbb{N}\setminus \{0\}$ of different markets or countries. In each market $j\in M$, each unit of the commodity is sold at a common price $r_j$, and the cost bore by company $i\in N$ for producing and/or transporting each unit of the commodity from its production-distribution centers to market $j$ is $c_{ij}$. We introduce $\alpha_{ij}: = r_j - c_{ij}$, for each $i\in N$ and $j\in  M$, and we assume w.l.o.g.\footnote{In production-distribution games, we may also consider situations where $\alpha_{ij}= r_j - c_{ij}< 0$ for some $i\in N$ and $j\in M$. Indeed there is no reason to serve any demand on market $j$ from a player $i$ if  $r_j < c_{ij}$, as we would be better off by not serving this demand. We can thus redefine any $c_{ij}$ such that $ c_{ij} > r_j$ by setting $c_{ij} = r_j$: if in an optimal solution to~(\ref{P2}), $y_{ij}$ is not zero for such a situation, one should simply remember to set it to zero in a post-processing phase.
} that each $\alpha_{ij}\geq 0$, i.e., $r_j\geq c_{ij}$, for each $i\in N$ and $j\in M$. Finally, each company $i\in N$ owns a part {$d_{ij}\in\setR_+$} of the total demand $d_j:=\sum_{i\in N} d_{ij}$ in market $j$, and has an upper bound $q_i\geq 0$ on its production capacity.  When $q_i> \sum_{i\in N} \sum_{j\in M} d_{ij}$ for all $i\in N$, we say that the production-distribution game is \emph{uncapacitated}. 
 We refer to the case where $m = 1$ as the \emph{single market} case, and to the case where $m > 1$ as the \emph{multi market} case.

We associate to the above setting a cooperative game $(N,v)$ that we call the {\em production-distribution game}. In the production-distribution game, the value function $v:2^N\mapsto \setR_+$ represents the maximum profit a coalition could make by reshuffling demands among its members. More precisely, if a subset $S\subseteq N$ of players collaborate, they achieve a total profit of $v(S)$ with:
\begin{eqnarray}
v(S) = \max \sum_{i\in S}\sum_{j\in M}  \alpha_{ij} y_{ij} & s.t. \label{P2}\\
\sum_{i\in S}y_{ij} \leq \sum_{i\in S}d_{ij} & \forall j\in M \nonumber \\
\sum_{j\in M}y_{ij} \leq q_i & \forall i\in S \nonumber \\
y_{ij} \geq 0 & \forall i\in S, j\in M \nonumber 
\end{eqnarray}

where $y_{ij}$ is the amount of commodity that company $i$ produces and distributes  for market $j$. 
Note that $v(S)$ is well-defined as it is a linear program whose set of feasible solutions is nonempty and bounded. It is nonempty because $\bar y=0$ is a feasible solution. It is  bounded because of the constraints $\sum_{i\in {S}}y_{ij} \leq \sum_{i\in {S}}d_{ij}$ and $y\geq 0$.

The development of this model was motivated by a collaboration between the third author and La Poste, France. La Poste operates three distinct parcel distribution branches — Chronopost (express delivery), DPD (for business clients), and Colissimo (for private customers) — which function independently. In certain urban areas, each branch typically operates its own sorting platform on the outskirts of the city, from which parcels are dispatched to various micro-hubs located in the city center (from which they are typically delivered to the final customers with cargo-bikes). A natural way to rationalize last-mile delivery would be to reshuffle parcels among branches before they enter the sorting platform, based on their final destination hub (see Fig.~\ref{fig1} for an illustration). However, despite belonging to the same parent group, each branch is required to demonstrate strong individual performance. As such, ensuring a fair allocation of the collaborative benefits is essential to gaining the support of the three branch managers required to implement such a solution. Note that the cost of reshuffling the demand among the sorting platforms is negligible compared to the cost of last-mile delivery in such a setting.

\begin{figure}[h]
    \centering
    \includegraphics[width=\textwidth]{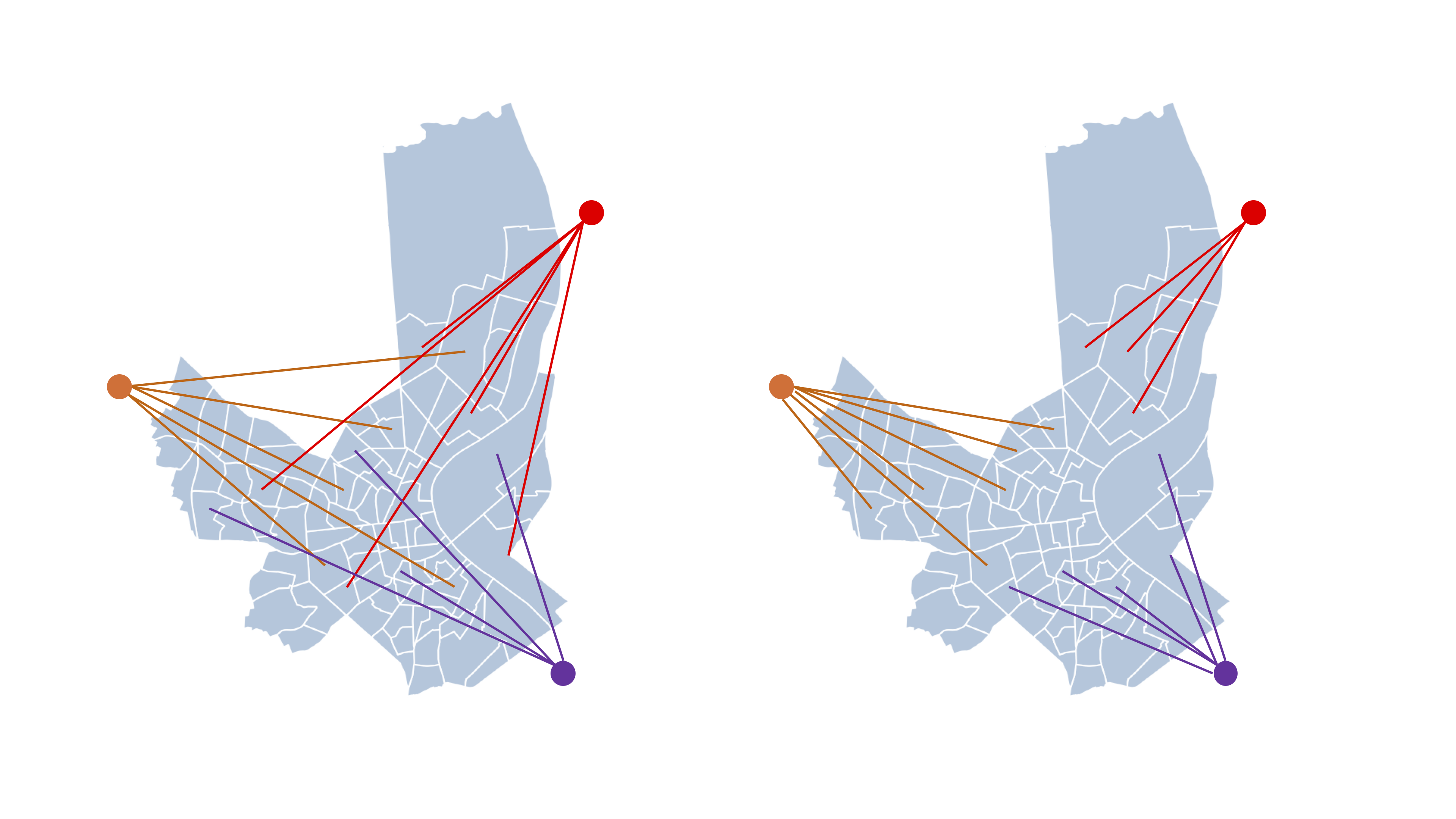}
    \caption{\footnotesize A stylized illustration of parcel flows in the city of Bordeaux, where the three companies (shown in red, brown, and purple) operate separate sorting platforms and serve different urban hubs. The left panel shows the situation when each company operates independently; the right panel illustrates the potential improvement if the companies collaborate. The total distance traveled is visibly shorter in the collaborative case.}
    \label{fig1}
\end{figure}

There are different standard solution concepts in cooperative game theory when it comes to splitting (transferable) utility fairly among the different actors involved in a coalition. The most popular are the core~\citep{gillies1959solutions}, the Shapley value~\citep{shapley1953value} and the nucleolus~\citep{schmeidler1969nucleolus}. In this work we are mainly interested in the core and the nucleolus.

\subsection*{Structure of the document}

In Section~\ref{lpg} we first show that the games under study constitute a special case of linear production games, which enables us to identify a core allocation through linear duality. We then turn to the uncapacitated setting for the study of the nucleolus. Section~\ref{updg} establishes several useful properties and specifies the conditions under which the core collapses to a singleton—the nucleolus. We subsequently design polynomial-time algorithms for computing the nucleolus when the number of markets is fixed. Each algorithm follows the Maschler scheme~\citep{maschler1979geometric}, summarised in Section~\ref{definition_nuc}. Section~\ref{singleuncap} shows how each iteration of the scheme can be executed in polynomial time using separation oracles and disjunctive reasoning, while Section~\ref{sec:comb} presents a dedicated and fast(er) combinatorial algorithm for the single market case and shows that in this case there exists a characterization set of linear length.

\section{Relation to linear production games}\label{lpg}

A well-established model for collaboration in production settings, known as \emph{linear production games}, was introduced and studied by Owen~\citep{owen1975core}. Linear production games is a subclass of \emph{totally balanced games}, first introduced by~\citet{shapley1969market}, a class of games in cooperative game theory that always admits a nonempty. Namely, a game $(N,v)$ is {\em balanced} if and only its core is nonempty (see below for the definition of core). The game is then {\em totally balanced} if also the subgame related to each coalition
of $N$ has a nonempty core: the subgame related to a coalition $S\subseteq N$ is the game $(S, v_S)$ where $v_S$ is the restriction of mapping $v$ to the subsets of $S$.

In Owen’s model, there are $n$ players, and each player possesses a certain amount of $r$ different resources. The resource vector of player $i$, for $i \in \{ 1,2,\ldots,n\}$, is given by $b^i = (b^i_1, b^i_2, \ldots, b^i_r)^\top$, where $b^i_k \geq 0$ denotes the amount of the $k$-th resource held by player~$i$. These resources can be used to produce $p$ different products, and each unit of product $j$, for $j\in \{ 1,2,\ldots,p$\}, can be sold at a fixed market price $c_j$. We denote the price vector by $c = (c_1, c_2, \ldots, c_p)$. Let $A = [a_{kj}]_{r \times p}$ be the linear production matrix\footnote{It is assumed that the production process is exogenous—that is, fixed in advance and independent of the identity or number of players involved.}, where $a_{kj}$ represents the amount of resource $k$ required to produce one unit of product~$j$ 
. The resulting linear production game is the pair $(N, v')$ defined as follows:

\begin{itemize}
  \item[(i)] The player set is $N = \{1, 2, \ldots, n\}$;
  \item[(ii)] For each coalition $S \subseteq N$, $v'(S)$ denotes the maximum profit that coalition $S$ can generate using the resources collectively held by its members, \emph{i.e.},
  \[
  v'(S) = \max \left\{ cx \;:\; Ax \leq \sum_{i \in S} b^i,\; x \geq 0 \right\}.
  \]
\end{itemize}

Production-distribution games can be formulated as linear production games, as stated below.

\begin{lemma}\label{newcor}
  Production-distribution games are linear production games.
\end{lemma}

\newcommand{\myproofblockA}{ 
 \ifnotshowproofs
  \subsection{Proof of Lemma~\ref{newcor}}
 \fi 
 \begin{proof}
    Consider a production-distribution game with $N$ players and $M$ markets, with demand $d_{ij}$, per-unit profit $\alpha_{ij}\geq 0$ for $i \in N$ and $j \in M$, and capacities $q_i$ for $i \in N$. We build an equivalent linear production game where the players are the elements in $N$, each pair $(i,j)$, $i \in N$ and $j \in M$, represents a product, and the elements in $N \cup M$ represent resources (production and demand resources, respectively). Each player $i \in N$ is endowed with a resource vector $b^i\in \setR^{n+m}$ with $b^i_i = q_i$ and $b^i_{i'} = 0$ for each $i' \neq i \in N$, and $b^i_{n+j} = d_{ij}$ for each $j \in M$.

    Consider in addition that each product $(i,j)$ consumes $a_{i,(i,j)} = 1$ unit of resource $i \in N$ and $a_{n+j,(i,j)} = 1$ unit of resource $j \in M$, and no other; and define $c_{(i,j)}=\alpha_{ij}$. The problem
    \[
    \max \left\{ cx \;:\; Ax \leq \sum_{i \in S} b^i,\; x \geq 0 \right\}
    \]
    is equivalent to
    \begin{eqnarray}
     \max \sum_{i \in N} \sum_{j \in M} \alpha_{ij} x_{ij} & \text{s.t.} \nonumber \\
    \sum_{i \in N} x_{ij} \leq \sum_{i \in S} d_{ij} & \forall j \in M \nonumber \\
    \sum_{j \in M} x_{ij} \leq q_i & \forall i \in S \nonumber \\
    \sum_{j \in M} x_{ij} \leq 0 & \forall i \not\in S \nonumber \\
    x_{ij} \geq 0 & \forall i \in N,\ j \in M \nonumber 
    \end{eqnarray}

    but the inequalities $\sum_{j \in M} x_{ij} \leq 0$ for all $i \not\in S$ and $x_{ij} \geq 0$ for all $i \not\in S,\ j \in M$ imply $x_{ij} = 0$ for all $i \not\in S$ and $j \in M$. The problem is thus equivalent to

    \begin{eqnarray}
     \max \sum_{i \in S} \sum_{j \in M} \alpha_{ij} x_{ij} & \text{s.t.} \nonumber \\
    \sum_{i \in S} x_{ij} \leq \sum_{i \in S} d_{ij} & \forall j \in M \nonumber \\
    \sum_{j \in M} x_{ij} \leq q_i & \forall i \in S \nonumber \\
    x_{ij} \geq 0 & \forall i \in S,\ j \in M \nonumber 
    \end{eqnarray}

    and the thesis follows.
\qed  \end{proof} 
}

\ifshowproofs
 \myproofblockA
\fi 

 An {imputation} is a vector $x\in \mathbb{R}^n$ with $x(N)=v(N)$. The core $\mathcal{C}$ of a game $(N,v)$ is the set of all  imputations $x\in \setR^n$ that ensures that no coalition has an incentive to break the grand coalition $N$, i.e., $\mathcal{C}$ is the set of $x$ such that:
	
\begin{eqnarray}
x(N)=v(N) & \nonumber \\
x(S)\geq v(S) &  \forall S\subseteq N \nonumber
\end{eqnarray}

Owen~\citep{owen1975core} proved that linear production games have nonempty core by exhibiting a solution from an arbitrary optimal dual solution to the linear program that determines $v(N)$. In the context of production-distribution games, the result reads as follows.  Consider the dual (D) of the linear program (P) defining $v(N)$:

\bigskip\noindent
\footnotesize
\begin{tabular}{cc}
\begin{minipage}[t]{0.58\textwidth}
\begin{tabular}{cc}
($P$) & $
\begin{array}{lll}
&\max \sum_{i\in N}\sum_{j\in M}\alpha_{ij}y_{ij}  \\
&\\
&\sum_{i\in N}y_{ij} \leq \sum_{i\in N}d_{ij} \ \ \ j\in M \nonumber\\
&\sum_{j\in M}y_{ij} \leq q_i \ \ \ \ \ \ \ \ \  \ \ \ \ \ i\in N \nonumber\\
&y_{ij} \geq 0 \ \ \  \ \ \ \  \ \ \ \ \ \  \ i\in N, j\in M \nonumber
\end{array}
$\\
\end{tabular}
\end{minipage}&
\begin{minipage}[t]{0.48\textwidth}
\begin{tabular}{cc}
($D$) & $
\begin{array}{ll}
& \min \sum_{i\in N} q_i\omega_i +  \sum_{j\in M} \beta_j \sum_{i\in N} d_{ij} \\
&\\
& \omega_i + \beta_j \geq \alpha_{ij} \ \ \  \ \ \ \ i\in N, j\in M  \\
& \omega_i \geq 0 \ \ \ \ \ \ \ \ \ \ \ \ \ \ \ \  \ \ \ \ \ \ \ \ \ \ \ \ i\in N \nonumber\\
&\beta_j \geq 0 \ \ \ \ \ \ \ \ \ \ \ \ \ \ \ \ \ \ \ \ \ \ \ \ \ \ \ j\in M \nonumber
\end{array}
$\\
\end{tabular}
\end{minipage}\\
\end{tabular}\\

\normalsize
\begin{lemma}\label{dual} For any production-distribution game, the dual problem $(D)$ has always an optimal solution. Moreover, if $(\omega^*, \beta^*)$ is an optimal solution to $(D)$, then the imputation $\gamma_i : = q_i\omega^*_i + \sum_{j\in M}  \beta^*_jd_{ij}$, for $i\in N$, is in the core. In particular, when the production-distribution game is uncapacitated, we have $w_i^*=0$ for all $i\in N$ and $\beta^*_j=\max_{i\in N} \alpha_{ij}$ for all $j\in M$. 
\end{lemma}

Since the problem of efficiently identifying a core solution is solved for production-distribution games, we turn our attention to the computation of the nucleolus (see Section~\ref{definition_nuc} for the definition). 

Computing the nucleolus is often NP-hard. As, for totally balanced games, positive results are limited to a few notable families, such as minimum-cost spanning tree games~\citep{granot1984core}, assignment games~\citep{solymosi1994algorithm}, and convex games~\citep{faigle2001computation}. While linear production games may appear structurally more tractable, \citet{deng2009finding} demonstrated that computing the nucleolus in this class remains NP-hard. 
However, the reduction presented in \citep{deng2009finding} does not apply directly to the case of production-distribution games. We leave open the question of whether the problem remains hard in the capacitated case or in the general uncapacitated setting. In this work, we focus on the uncapacitated case with a small number of markets, including, of course, the single market case.
\begin{remark}

It is known that the nucleolus can be computed in polynomial time for the class of convex games~\citep{faigle2001computation}. A (profit) game $(N,v)$ is convex if $v(S)+v(T)\leq v(S\cap T)+v(S\cup T)$ for all coalitions $S$ and $T$. Unfortunately, production-distribution games do {\em not} fit into the framework of convex games, not even in the case of uncapacitated single market, as it is witnessed by the following example. 

\begin{example}
Consider the uncapacitated production-distribution game with 3 players and a single market with the following input: $r=2$, $c_1=1,c_2=1,c_3=2$, $d_1=d_2=d_3=1/3$. Consider the set $S=\{1,3\}$, $T=\{2,3\}$. We have $v(S)+v(T)=2/3+2/3=4/3>0+1=v(S\cap T)+v(S\cup T)$.
\end{example}

\end{remark}

\section{Uncapacitated production-distribution games}\label{updg}

 For uncapacitated production-distribution games, the capacity constraints $\sum_{j\in M} y_{ij} \leq q_i$, $i\in N$ are redundant.  We may assume in addtition, without loss of generality, that there is equality in the constraint $\sum_{i\in S}y_{ij} \leq \sum_{i\in S}d_{ij}$ in~\eqref{P2}. Indeed, if not, we could increase the value of some $y_{ij}$ so that the constraint holds tight and get a solution with objective value no smaller, since $\alpha_{ij}\geq 0$, for any $i\in N$ and $j\in M$. In the uncapacitated case, we may therefore re-write~\eqref{P2} as follows:

\begin{eqnarray}
v(S) = \max \sum_{i\in S}\sum_{j\in M} \alpha_{ij}y_{ij} & s.t. \label{P3}\\
\sum_{i\in S}y_{ij} = \sum_{i\in S}d_{ij} & \forall j\in M \nonumber \\
y_{ij} \geq 0 & \forall i\in S, j\in M \nonumber 
\end{eqnarray}

and it is therefore straightforward to see that:

\begin{equation}\label{vuncap}
v(S) = \sum_{j\in M} (\sum_{i\in S} d_{ij}) \cdot \max_{k\in S} \alpha_{kj}.  
\end{equation}

Note that when the core consists of a single point, the nucleolus (which is in the core when the core is nonempty) can be directly computed using Lemma~\ref{dual}. A well-known sufficient condition for a game $(N, v)$ to have a core that is nonempty and consists of a single point is the following: $v(A\cup B) = v(A) + v(B)$, for any $A, B\subset N: A \cap B = \emptyset$. Such a game is called {\em inessential}, meaning that forming a coalition is not beneficial.  

Interestingly, we can {\em characterize}, i.e., through conditions that are necessary and sufficient, the cases in which the core of the uncapacitated production-distribution game reduces to a single point, namely, the nucleolus. It turns out that these conditions, that are specified in the following lemma, are both relevant in practical settings, as competitive dynamics may naturally enforce them, and from a theoretical point of view since they hold also in games that are not inessential and where players have indeed the incentive to cooperate.

Lemma~\ref{lem:singleton} will provide this caracterization that builds upon the following remark (see the Appendix for details):

\begin{remark}\label{decomposition}
In the uncapacitated case, it is possible to associate to the multi market game $(N,v)$ a single market (sub-)game $(N, v_j)$ for each market $j\in M$, so that $(N, v)=(N,\sum_{j\in M} v_j)$. Namely, if we let $v_j(S) = (\sum_{i\in S} d_{ij} )\cdot  a^S_j$, for each $S\subseteq N$, with $a^S_j=\max_{k\in S} \alpha_{kj}$, then $v(S) = \sum_{j\in M}v_j(S)$. 
\end{remark}

\begin{lemma}\label{lem:singleton}
Let $(N,v)$ be an uncapacitated production-distribution game and  let $v_j(S) = (\sum_{i\in S} d_{ij} ) \cdot  \alpha^S_{j}$, for each $j\in M$ and each $S\subseteq N$. The following assertions are equivalent:
\begin{itemize}
\item[(i)] The core of $(N,v)$ is a singleton
\item[(ii)] For each $j\in M$, the core of $(N,v_j)$ is a singleton
\item[(iii)] For all $j\in M,i\in N,k\in N\setminus \{i\}$, we have $ d_{kj} \cdot(\alpha^N_j - \alpha^{N\setminus \{i\}}_j)=0$.
\end{itemize} 
\end{lemma}

\newcommand{\myproofblockB}{ 
 \ifnotshowproofs
  \subsection{Proof of Lemma~\ref{lem:singleton}}
 \fi
\begin{proof}
The Minkowski sum of the cores of the sub-games $(N, v_j)$, which are nonempty, is included in the core of the multi market game and hence we have $(i)\implies (ii)$.

Let us prove that $(ii)\implies (iii)$ by contraposition. Assume there exist $j\in M, i\in N$ and $k\in N\setminus \{i\}$ such that $ d_{kj} \cdot(\alpha^N_j - \alpha^{N\setminus \{i\}}_j)\neq 0$. Then  we have in particular $|\argmax_{k\in N} \alpha_{kj}|=1$ and $i$ is the unique maximizer,  otherwise $\alpha^N_j =\alpha^{N\setminus \{i\}}_j$. We also have $d_{kj}>0$, with $k\neq i$. Now observe that the point $\bar{x}$ such that $\bar{x}_l:=d_{lj}\alpha_{ij}$, for each $l\in N$, is in the core of $(N,v_j)$ by Lemma~\ref{dual}. For each $l\in N$, define:
$$
\tilde{x}_l=\left\lbrace
\begin{array}{ll}
\bar{x}_l +\epsilon & \mbox{if } l=i,\\
\bar{x}_l-\epsilon & \mbox{if }l={k},\\
\bar{x}_l & \mbox{otherwise.}
\end{array}
\right.
$$
with $\epsilon=min_{S:i\not\in S,k\in S} (\bar{x}(S) - v_j(S))$. Observe that $\epsilon>0$ as for $S$ with $i\not\in S$ and $k\in S$, we have $\bar{x}(S)=(\sum_{l\in S} d_{lj})\cdot \alpha_{ij} > (\sum_{l\in S} d_{lj})\cdot \max_{l\in S} \alpha_{lj} =v_j(S)$ as $\sum_{l\in S} d_{lj} \geq d_{kj}>0$. 

$\tilde{x}$ is an imputation as $\tilde{x}(N)=\bar{x}(N)=v(N)$. Now, for $S$ such that $i\in S$ or $k\not\in S$, we have $\tilde{x}(S)\geq \bar{x}(S)\geq v_j(S)$. For $S$ such that $i\not\in S$ and $k\in S$ we have $\tilde{x}(S)=\bar{x}(S)-\epsilon\geq v_j(S)$ and thus $\tilde{x}\neq \bar{x}$ is in the core of $(N,v_j)$. Hence the core of $(N,v_j)$ is not a singleton.

We now prove that $(iii)\implies (i)$. For each $i\in N$, we have $$v(N\setminus\{i\})+\sum_{j\in M}  d_{ij} \alpha^N_j =\sum_{j\in M} \sum_{k\in N\setminus \{i\}}  d_{kj} \alpha_j^{N\setminus \{i\}} +\sum_{j\in M}  d_{ij} \alpha^N_j =\sum_{j\in M} \sum_{k\in N\setminus \{i\}}  d_{kj} \alpha_j^{N} +\sum_{j\in M} d_{ij} \alpha^N_j =v(N)$$.

It follows that any point $x$ in the core of $(N,v)$ must satisfy the following constraints:

\begin{eqnarray}
\label{core=1}
x(N)&=&v(N)\\
\label{core-ineq1}
x(N\setminus\{i\})\geq v(N\setminus \{i\})&= &v(N) - \sum_{j\in M} d_{ij} \alpha^N_j, \ \  i\in N\\
\end{eqnarray}
Summing up inequalities (\ref{core-ineq1}) we get:

$$(n-1) x(N) \geq n \cdot v(N) - v(N) = (n-1) v(N).$$


Therefore, any solution in the core must satisfy each inequality~\eqref{core-ineq1} tight. As the system induced by these inequalities has full rank, it follows that it has unique solution and therefore the core, which is nonempty, has a single point.

\qed  \end{proof} 

}

\ifshowproofs 
 \myproofblockB
\fi

 It might be tempting to try to exploit Remark~\ref{decomposition} further to decompose the computation of the nucleolus. Unfortunately, the nucleolus is not additive in general (in contrast with other solution concepts, e.g. the Shapley value~\citep{shapley1953value}), not even for uncapacited production-distribution games, as it is shown by the following example.

\begin{example}\label{theexample}
Consider an uncapacitated game with 3 producers and 3 markets. The demand matrix $d:=(d_{ij})_{1\leq i,j\leq 3}$ and the profit matrix $\alpha:=(\alpha_{ij})_{1\leq i,j\leq 3}$, are given below: \\

$
d=\left[\begin{array}{ccc}
 1 & 0 & 1\\
 0 & 1 & 1\\
 1 & 1 & 0\\
\end{array}\right]
$
, 
$
\alpha=\left[\begin{array}{ccc}
 1 & 1 & 1 \\
1 & 0 & 0 \\
0 & 0 & 1 \\
\end{array}\right]
$\\

Computation shows that the nucleolus of the associated production-distribution game  is $x_1=10/3,x_2=4/3,x_3=4/3$. In contrast, the sum of the nucleolus associated with the 3 production-distribution games restricted to each market is  $x_1=3,x_2=1.5,x_3=1.5$. 
\end{example}

\section{The nucleolus and the Maschler scheme}\label{definition_nuc}

For each $x\in\setR^n$ such that $x(N)=v(N)$ and for each $S \in 2^N$, we define the excess of coalition $S$ at $x$ as: $e_S(x) = x(S)- v(S)$. We also define the vector $e(x) \in \setR^{2^n}$ of the excess of the different coalitions at $x$ as $e(x) = (e_S(x))_{S\in 2^N}$. We say that $e(x) \succ e(y)$ if $\overrightarrow{e(x)}$ is lexicographically superior to $\overrightarrow{e(y)}$, where for a vector $x$, $\overrightarrow{x}$ is a permutation of the entries of $x$ arranged in non-decreasing order. The (pre)nucleolus is the set of imputations that lexicographically maximizes the excess vector. It is known to be a singleton~\citep{schmeidler1969nucleolus}. We often abuse notations and identify the (pre)nucleolus with the unique element in the set. Note that, in principle, the definition of the nucleolus imposes individual rationality but this is redundant when the core is nonempty~\citep{peleg2007introduction}, as in production-distribution games. In this case, we have: 

\begin{definition}\label{nucleolus}: The nucleolus of a cooperative game $(N, v)$ with nonempty core is $nc(N, v) := \{x\in\setR^n : x(N)=v(N)$ and $\nexists y \in\setR^n : y(N)=v(N) , e(y) \succ e(x)\}$.\end{definition}

One of the main computational technique to compute the nucleolus is known as Maschler scheme~\citep{maschler1979geometric}~: it reduces the problem to solving a sequence of linear programs. There are different variants of the Maschler scheme. We present here a somewhat generic version that will be instantiated in two different ways in the following sections.

\smallskip
The Maschler scheme computes the nucleolus $\nu$ through a sequence of linear programs. In each iteration, it fixes the total payoff for certain coalitions $S$ at the value $\nu(S)$—as determined in previous steps from optimality conditions—even though the individual components $\nu_i$ for $i \in S$ may still be unknown. It then constructs a new linear program aimed at maximizing the smallest excess across all remaining, unfixed coalitions. This procedure continues until the constraints define a unique solution: the nucleolus.




More formally, we solve iteratively, for $k\geq 1$, the following problem, until the optimal solution set is made of a singleton. 

\begin{align}
&\mbox{max } \epsilon \nonumber \\
(M^k) \ \ \ \ & x(S)\geq v(S) + \epsilon \ \ \  \forall S\in 2^N \setminus \mathcal{F}^k \label{coreconsk} \\
&x(S) = c^k(S) \ \ \ \ \ \  \forall S\in\mathcal{F}^k. \nonumber
\end{align}


By induction problem $(M^k)$ is bounded and feasible and we let $\epsilon^k$ be the optimal value. $\mathcal{F}^k$ represents the set of coalitions $S$ for which we fix the value of $x(S)$ to $c^k(S)$ for some value $c^k(S)$, with $x(S)=c^k(S)$ being satisfied by the nucleolus. Initially
 $\mathcal{F}^1=\{N,\emptyset\}$ and $c^1(N)=v(N), \ c^1(\emptyset)=0$. Then, for $k\geq 2$, we define $\mathcal{F}^{k} $ by adding some coalitions $S$ in $2^N \setminus \mathcal{F}^{k-1}$  with $y(S)=z(S)=:\rho(S)$ for any two optimal solution $(y,\epsilon^{k-1})$,  $(z,\epsilon^{k-1})$ to $(M^{k-1})$.  The nucleolus $\nu$ being on the optimal face, it satisfies $\nu(S)=\rho(S)$. We define $c^k(S)=c^{k-1}(S)$ for $k\geq 2$ and $S\in \mathcal{F}^{k-1} \cap \mathcal{F}^{k}$ and $c^k(S)= \rho(S)$ for $S\in \mathcal{F}^{k} \setminus \mathcal{F}^{k-1}$. The  constraints $x(S) = c^k(S)$ for $S\in \mathcal{F}^k$ are satisfied by the nucleolus by induction.
 
The original scheme considered by Maschler is when $\mathcal{F}^{k}=\mathcal{F}^{k-1}\cup \mathcal{S}^{k-1}
$, where $\mathcal{S}^{k-1}$ coincides with the set of all coalitions $S$ such that $y(S)=z(S)$ for any two optimal solutions $(y,\epsilon^{k-1}),(y,\epsilon^{k-1})$ to ($M^{k-1}$). In this case, it can be proven that the dimension of the polytope associated with the solution set to ($M^k$) decreases at each iteration and thus that the number of iterations is bounded by $n$~\citep{maschler1979geometric}.

 Note that we could also in principle remove some coalitions from $\mathcal{F}^{k-1}$ in $\mathcal{F}^{k}$  (except for $N$) in case we want to have specific structure for $\mathcal{F}^{k}$. We could also add some constraints that are implicit (linear combinations of others). What matters is that we only set constraints of the form $x(S) = c^k(S)$ for $S\in \mathcal{F}^k$ that are satisfied by the nucleolus and that problem $(M^k)$ is still bounded (this is the reason why $N$ must be kept). When such variants are considered, one needs to specify $\mathcal{F}^{k}$, to prove convergence, and to evaluate the number of iterations explicitly. Note that if this generic procedure converges, it converges by construction to the nucleolus.

\section{Separation algorithms for the nucleolus of the uncapacitated production-distribution game}\label{singleuncap}

In this section we provide a  separation algorithm to solve each iteration of the Maschler scheme  and find the nucleolus in polynomial time when the number of markets is fixed. We begin with the case of a single market and then extend the approach to the general setting. We consider an instantiation of the Maschler scheme presented in Section~\ref{definition_nuc}  where $\mathcal{F}^k=\{S: S\in span\{S_1,...,S_{k}\}\}$ for some $S_1,...,S_{k} \subseteq N$ that are linearly independent (we are abusing notations here and in the following as we identify sets with their characteristic vectors). In this case, the set of constraints $x(S) = c^k(S)$ for all $S\in \mathcal{F}^k$, can be reduced to $x(S) = c^k(S)$, $\forall S: S \in \{S_1,...,S_{k}\}$ as all other equalities are dominated (remember the system is always feasible by construction).  

At each iteration $k\geq 1$, we thus require the solution of a linear program of the kind:
\begin{align}
&\mbox{max } \epsilon \nonumber \\
(M^k) \ \ \ \ & x(S)\geq v(S) + \epsilon \ \ \  \forall S: S \not\in span\{S_1,...,S_{k}\} \label{reveq1}\\
&x(S) = c^k(S) \ \ \ \ \ \  \forall S: S \in \{S_1,...,S_{k}\} \label{reveq2}
\end{align}
where, once again each $c^k(S)$ in~\eqref{reveq2} is suitably defined through the iterations. At the first iteration, $S_1=N$ and $c^1(N)=v(N)$. Then the set $S_{k+1}$ that is added to $\{S_1,...,S_{k}\}$ from iteration $k$ to iteration $k+1$ corresponds to a single constraint in (\ref{reveq1}) with positive value in an optimal dual solution to $(M^k)$. Such a set exists as an optimal dual solution $\mu$ to $(M^k)$, with $\mu_S\geq 0$ for $S\not\in span\{S_1,...,S_{k}\}$, must contain at least one positive dual variable over the sets in~\eqref{reveq1} as, in the dual, there is a constraint of the form: $\sum_{S: S \not\in span\{S_1,...,S_{k}\}} \mu_S =1$. 

By induction, at each iteration, $\{S_1,...,S_{k}\}$ is made of linearly independent sets and so $k\leq n$. The corresponding
 scheme will thus converge in at most $n$ iterations. Moreover, we can state:

\begin{lemma}\label{polynomial}
If each linear program $(M^k)$ can be solved in polynomial time and still in polynomial time a dual certificate associated with facets of (\ref{reveq1})-(\ref{reveq2}) can be found, then the nucleolus can be found in polynomial time.
\end{lemma}

\subsection{Single market}\label{newsepalg}

In the single market game, we simplify notations and assume the following. The commodity product is sold at market price $r$ and the (per unit) production cost for company $i$ is $c_i$. Each company $i$ owns a demand $d_i$. We let $\alpha_i:= r-c_i$ and we assume here without loss of generality that $\alpha_1\geq \alpha_2 \geq ... \geq \alpha_n \geq 0$. The value $v(S)$ of a coalition $S\subseteq N$ in the single market game is therefore equal to:

\begin{equation}\label{sinvuncap}
v(S) = (\sum_{i\in S} d_i) \cdot \max_{k\in S}\alpha_k.
\end{equation}

 We now show how to separate any point from the polytope $(M^k)$ in polynomial time for the single market uncapacitated production-distribution game. This allows to exploit Lemma~\ref{polynomial} to show that the nucleolus can be computed in polynomial time.

Let us introduce a couple of notations we will use in the following. For a set $S\subseteq E$ and $e\in E$, we denote by $S-e$ the set  $S\setminus\{e\}$ and by $S+e$ the set $S\cup \{e\}$.  A set $\{a,...,b\}$, with $a, b\in \{1,..,n\}$ and $b<a$, is the empty set (not the set $\{b,...,a\}$).

\begin{lemma}\label{lem:sepk}
There exists a polynomial time algorithm for separating over the polytope $(M^k)$ for the single market uncapacitated production-distribution game.
\end{lemma}

\newcommand{\myproofblockC}{
 \ifnotshowproofs
  \subsection{Proof of Lemma~\ref{lem:sepk}}
 \fi 
\begin{proof}
Let $(\bar{x},\bar{\epsilon})$ be a point in $\setR^{n+1}$ such that $\bar{x}(N)=v(N)$ and $\bar{x}(S)=c^k(S)$ for all $S\in \{S_1,...,S_{k}\}$.  We want to decide whether  $\bar{x}$ satisfies inequalities (\ref{reveq1}) and in case not, find a separating inequality.

The question reduces to testing whether there exists $S\subseteq \{1,...,n\}$, with $S\not\in span\{S_1,...,S_{k}\}$, such that $\bar{x}(S) < v(S) + \bar{\epsilon}$. This is equivalent to checking for each $i=1,...,n$ and all $S$ s.t. $i\in S$, $S\cap\{1,...,i-1\}=\emptyset$ and $S\not\in span\{S_1,...,S_{k}\}$ whether  $\bar{x}(S) \geq \demand(S)\cdot \alpha_i + \bar{\epsilon}$. Setting $\bar{y}^{i}_j:= \bar{x}_j -  \demand_j\cdot \alpha_i$, for any $j\in \{i,...,n\}$, this is equivalent to testing whether  $\bar{y}^{i}(S) \geq \bar{\epsilon}$ or equivalently whether $\bar{y}^{i}(S-i)  \geq \bar{\epsilon} - \bar{y}^{i}(\{i\})$. We state this question in a slightly more abstract form that will be used also in the multi market case, see Section~\ref{umm}. Namely, for each $i<n$ (for $i=n$ the problem is trivial), the problem is of the following form: 

\begin{question}\label{enumsep}
Consider a set $R$ with $|R|=p$ for some $p\in \mathbb{N}\setminus \{0\}$, a set $Q$ with $Q\cap R=\emptyset$ and linearly independent sets $S_1,...,S_k$ over a ground set of cardinality $n$ containing $Q\cup R$ for some $0 \leq k\leq n$. Let $a\in \mathbb{R}^p$ and $b$ in $\setR$, and assume w.l.o.g. that the elements in $R$ are numbered $\{1,...,p\}$ with $a_1 \leq ... \leq a_p$. Check whether $a(S) \geq b$ for all $S \subseteq R:  S\cup Q \not\in span\{S_1,...,S_{k}\}$ and identify $S$ with $a(S)<b$ if not. 
\end{question}

It is indeed enough to set:  $R=N\setminus \{1,...,i\}$ ; $p=n-i$ ; $Q=\{i\}$ ; $a=\overrightarrow{y_{|R}}$ ; $b=\bar{\epsilon} - \bar{y}^{i}(\{i\})$ (recall that for a vector $x$, $\overrightarrow{x}$ is a permutation of the entries of $x$ arranged in non-decreasing order). 

We now show how to solve Question~\ref{enumsep}. From now on, we let $\mathcal{F}:=\{S \subseteq \{1,...,p\}: S\cup Q \in span\{S_1,...,S_{k}\}\}$. Also we say that the inequality associated with $S\subseteq \{1,...,p\}$ is {\em violated} if $a(S) < b$ (independently of whether $S$ is in $\mathcal{F}$ or not). In contrast, we talk of a {\em separating inequality} for a violated inequality associated with a set $S$ not in $\mathcal{F}$.

We define $\bar{j}:=\max\{j:a_j\leq 0\}$, with the convention that $\bar{j}=0$ if $a_1> 0$, and $\bar{S}:=  \{1,...,\bar{j}\}$ and we assume also that $a(\bar{S}) < b$ and $\bar{S}\in \mathcal{F}$ otherwise either there is no violated inequality -- the one associated with $\bar{S}$ is the most (possibly) violated one as $\bar{S}$ minimizes $a(S)$ over all $S\subseteq \{1,...,p\}$ -- or  the inequality associated with $\bar{S}$ is a separating inequality.

\vspace{1ex}

\noindent If $\bar{j}>0$ (i.e., $a_1 \leq 0$) we let:
\begin{equation} \underline{l}:=\min\{l \in \{1,..., \bar{j}\}: \bar{S} - l' \in \mathcal{F} \mbox{ and } a(\bar{S}-l')<b, \forall l' : l\leq l' \leq \bar{j}\} 
\label{eq:l}
\end{equation}

\noindent and analogously if $\bar{j}<p$ (i.e., $a_p > 0$) we let:
\begin{equation} \bar{l}:=\max\{l \in \{\bar{j}+1,...,p\}: \bar{S} + l' \in \mathcal{F} \mbox{ and } a(\bar{S}+l')<b, \forall l' : l\geq l' \geq \bar{j}+1\}
\label{eq:u}
\end{equation}

\vspace{1ex}

with respectively the convention that $\underline{l}=\bar{j}+1$ if there is no $l\in \{1,..., \bar{j}\}$ with the  property required in (\ref{eq:l}) and that $\bar{l}=\bar{j}$ if there is no $l\in \{\bar{j}+1,...,p\}$ with the  property required in (\ref{eq:u}). 

Now suppose that $\underline{l}> 1$ (we are therefore assuming that $\bar{j}>0$). By definition, $\bar{S} - (\underline{l}-1) \not\in \mathcal{F}$ or $a(\bar{S}-(\underline{l}-1))\geq b$. So, if in particular $a(\bar{S} -  (\underline{l}-1))< b$, then $\bar{S} - (\underline{l}-1) \not\in \mathcal{F}$ and the inequality associated with $\bar{S} -  (\underline{l}-1)$ is a separating inequality. Analogously, suppose that $\bar{l}<p$ (we are therefore assuming that $\bar{j}<p$). By definition, $\bar{S} + (\bar{l}+1) \not\in \mathcal{F}$ or $a(\bar{S}+(\bar{l}+1))\geq b$. So, if in particular $a(\bar{S}+ (\bar{l}+1))< b$, then $\bar{S} + (\bar{l}+1) \not\in \mathcal{F}$ and the inequality associated with $\bar{S}+ (\bar{l}+1)$ is a separating inequality. 

We claim that if we are not in one of the two situations above, there is no separating inequality. Before going into details, observe that we are in (exactly) one of the following three cases:

\begin{enumerate}
    \item $0<\bar{j}<p$; either $\underline{l}=1$ or $a(\bar{S} - (\underline{l}-1))\geq b$; either $\bar l = p$ or $a(\bar{S}+ (\bar{l}+1))\geq b$;
    \item $\bar{j}=p$ (hence $\underline{l}$ is defined while $\bar{l}$ is not); either $\underline{l}=1$ or $a(\bar{S} - (\underline{l}-1))\geq b$;
    \item $\bar{j}=0$ (hence $\bar{l}$ is defined while $\underline{l}$ is not);  either $\bar l = p$ ir $a(\bar{S}+ (\bar{l}+1))\geq b$.
\end{enumerate}

Now we first claim that:
\begin{itemize}
    \item[(i)] for cases 1 and 2, any violated inequality must include all $l$ in $\{1,...,\underline{l}-1\}$;
    \item[(ii)] for cases 1 and 3, any violated inequality must miss all $l$ in $\{\bar{l}+1,...,p\}$.
\end{itemize}
Note that $(i)$ is trivial if $\underline{l}=1$; analogously $(ii)$ is trivial if $\bar{l} = p$. So from now we assume that $\underline{l}> 1$ and $\bar{l}< p$. 
Now observe that a most possibly violated inequality missing (at least) an element $l\in \{1,...,\underline{l}-1\}$ is the inequality associated with $\bar{S}-(\underline{l}-1)$ and the most possibly violated inequality including (at least) an element $l\in \{\bar{l}+1,...p\}$ is the inequality associated with $\bar{S}+ (\bar{l}+1)$, which are both not violated. Hence our claim holds.

Suppose now that case 1 holds (the other cases go along the same lines). Any violated inequality, if any, is  associated with a set of the form $\{1,...,\underline{l}-1\} \cup \tilde{S}$ for some $\tilde{S} \subseteq \{\underline{l},...,\bar{l}\}$. As we now show, each $l\in \{\underline{l},...,\bar{l}\}$ is such that  $\{l\} \in span\{S_1,...,S_{k}\}$: by construction, for each $l:\underline{l} \leq l \leq \bar{j}$, $\bar{S} - l \in \mathcal{F}$ and $\bar{S}\in \mathcal{F}$, so $\{l\} = \bar{S}\setminus (\bar{S} - l)=(\bar{S}\cup Q) \setminus ((\bar{S} - l)\cup Q)$ is in $span\{S_1,...,S_{k}\}$ too; analogously,  by construction, for each $l: \bar{l} \geq l \geq \bar{j}+1$, $\bar{S}+l \in \mathcal{F}$ and $\bar{S}\in \mathcal{F}$, so $\{l\} = (\bar{S}+l) \setminus \bar{S} = ((\bar{S}+l) \cup Q)\setminus (\bar{S}\cup Q)$ is in $span\{S_1,...,S_{k}\}$ too. Since each $l\in \{\underline{l},...,\bar{l}\}$ is such that  $\{l\} \in span\{S_1,...,S_{k}\}$, it follows that $(\bar{S}\setminus (\{\underline{l},...,\bar{j}\}\setminus \tilde{S})) \cup (\tilde{S} \cap  \{\bar{j}+1,..., \bar{l}\})= \{1,...,\underline{l}-1\} \cup \tilde{S}$ is in $\mathcal{F}$.
\qed  \end{proof} 

}

\ifshowproofs 
 \myproofblockC
\fi

\begin{corollary}\label{cor:poly}
The nucleolus can be found in polynomial time for the single market uncapacitated production-distribution game.
\end{corollary}

\newcommand{\myproofblockD}{ 
 \ifnotshowproofs
   \subsection{Proof of Corollary~\ref{cor:poly}}
 \fi 
\begin{proof}
We build upon Lemma~\ref{polynomial}. We can solve each linear program $(M^k)$ in polynomial time from the equivalence between optimization and separation. Besides, we can obtain in polynomial time a dual certificate associated with facets of (\ref{reveq1})-(\ref{reveq2}) (see from Corollary 14.11.g in~\citep{schrijver1998theory}). 
\qed  \end{proof} 
}
\ifshowproofs 
 \myproofblockD
\fi

\subsection{Fixed number of markets}\label{umm}

We now deal with the multi market case, i.e., $|M|\geq 1$. Each unit of the commodity is sold at a common price $r_j$ in each market $j\in M$, the (per unit) production cost of company $i$ in market $j$ is $c_{ij}$, and therefore $\alpha_{ij}=r_j-c_{ij}$ is the (per unit) profit of company $i$ in market $j$. Moreover, each company $i$ owns a demand $d_{ij}\geq 0$ in market $j$ and so the value $v(S)$ of a coalition $S\subseteq N$ is equal to (see~\eqref{vuncap}):

$$v(S) = \sum_{j\in M} (\sum_{i\in S} d_{ij}) \cdot \max_{k\in S}\alpha_{kj}.$$

While we believe that Remark~\ref{decomposition} suggests that the sum of the nucleolus of the single market subproblems might be a good solution concept for this multi market case (as it seems reasonable to argue on each market separately, see also Section~\ref{conclusion}), Example~\ref{theexample} shows that the nucleolus is {\em not} additive in general. We will however show that the nucleolus can be computed in polynomial time, in the multi market case, when $|M|$ is fixed.

 We in particular show how to separate any point from the polytope $(M^k)$, defined by (\ref{reveq1}) and (\ref{reveq2}), in polynomial time for the multi market uncapacitated production-distribution game with a fixed number of markets. This allows to exploit again Lemma~\ref{polynomial} and Corollary~\ref{cor:poly} to show that the nucleolus can be computed in polynomial time in this case.

\begin{lemma}\label{lem:sepfixed}
There exists a polynomial time algorithm for separating over the polytope $(M^k)$ for the multi market uncapacitated production-distribution game with a fixed number of markets.
\end{lemma}

\newcommand{\myproofblockE}{ 
 \ifnotshowproofs
   \subsection{Proof of Lemma~\ref{lem:sepfixed}}
 \fi 
\begin{proof}
     The separation problem is : given a point $(\bar{x},\bar{\epsilon})$   in $\setR^{n+1}$ such that $\bar{x}(N)=v(N)$ and $\bar{x}(S)=c^k(S)$ for all $S\in \{S_1,...,S_{k}\}$; decide whether  $\bar{x}$ satisfies inequalities (\ref{reveq1}) and in case not, find a separating inequality.
 We can apply a similar rationale as for the single market case to decompose into simpler problems. 
 
 We may associate with each (nonempty) $S$ with $S\not\in span\{S_1,...,S_{k}\}$ an \emph{index vector} $(i_1,...,i_m)$ such that $i_j\in \argmax_{i\in S} \alpha_{ij}$. There are of course at most $n^{m}$ possible index vectors. Now, for every $j\in M$ and $l\in N$, define the set $\mathcal{S}^l_j:=\{i\in N:\alpha_{ij} \leq \alpha_{lj}\}$. Then, for each candidate index vector $(i_1,...,i_m)$, we consider the sets $S$ that are consistent with it, meaning $S\not\in span\{S_1,...,S_{k}\}$, $\{i_1,...,i_m\}\subseteq S$ and $S\subseteq S^{i_j}_j$ for all $j\in M$. The separation problem thus decomposes  into  $n^m$ problems, one for each candidate index vector $(i_1,...,i_m)$,  of the following form:\\

\noindent {\em 
$\exists ? \ S\not\in span\{S_1,...,S_{k}\}$ s.t. $\{{i}_1, ...,{i}_m\}\subseteq S\subseteq \cap_{j\in M}  \mathcal{S}^{i_j}_j$ with $\bar{x}(S) <\sum_{j\in M} \alpha_{i_j,j} \cdot \sum_{i\in S} d_{ij} + \bar{\epsilon}$.}\\

 \noindent Note that some of the problems might be unfeasible (when  there exists $j$ such that $i_j\not \in \cap_{j\in M} \mathcal{S}^{i_j}_j$ ) and we simply discard those.  The feasible problems are of the form: \\
 
 \noindent{\em  $\exists ? S: \ S\cup\{i_1,...,i_m\} \not\in span\{S_1,...,S_{k}\}$ s.t. $S\subseteq \cap_{j\in M} \mathcal{S}^{i_j}_j \setminus \{{i}_1, ...,{i}_m\}$ with 
 $\sum_{i\in S} \bar{y}_i < \bar{\epsilon} -\sum_{j\in [m]} \bar{y}_{i_j}$ }\\
 
 \noindent where $\bar{y}_i=\bar{x}_i - \sum_{j\in M} \alpha_{i_j,j} \ d_{ij}$.  It follows that the question of recognizing whether $\bar{x}(S) < v(S) + \bar{\epsilon}$, for a set $S\not\in span\{S_1,...,S_{k}\}$, can be handled again by the more general question discussed in Section~\ref{newsepalg} and recalled in the following:

\vspace{2ex}
\noindent {\em Question 1.}\ 
\ Consider a set $R$ with $|R|=p$ for some $p\in \mathbb{N}\setminus \{0\}$, a set $Q$ with $Q\cap R=\emptyset$ and linearly independent sets $S_1,...,S_k$ over a ground set of cardinality $n$ containing $Q\cup R$ for some $0 \leq k\leq n$. Let $a\in \mathbb{R}^p$ and $b$ in $\setR$, and assume w.l.o.g. that the elements in $R$ are numbered $\{1,...,p\}$ with $a_1 \leq ... \leq a_p$. Check whether $a(S) \geq b$ for all $S \subseteq R:  S\cup Q \not\in span\{S_1,...,S_{k}\}$ and identify $S$ with $a(S)<b$ if not. 
\vspace{2ex}

It is indeed enough to set for each $(i_1,...,i_m)$ such that $\{i_1,...,i_m\} \subsetneq \cap_{j\in M} \mathcal{S}^{i_j}_j  $ (the case where $\{i_1,...,i_m\} = \cap_{j\in M} \mathcal{S}^{i_j}_j  $ is trivial) : $Q=\{i_1,...,i_m\}$; $R=\cap_{j\in M} \mathcal{S}^{i_j}_j \setminus Q$ ; $p=|R|$; $a=\overrightarrow{ \bar{y}_{|R}}$, ; $b=\bar{\epsilon} -\sum_{j\in [m]} \bar{y}_{i_j}$. Since we have shown that Question~\ref{enumsep} can be solved in polynomial time, it follows that a separating inequality $\bar{x}(S) < v(S) + \bar{\epsilon}$ (if any) associated with a set $S\not\in span\{S_1,...,S_{k}\}$ can be found in polynomial time when $|M|=m$ is fixed by enumeration over all possible $(i_1,...,i_m)$. 
\qed  \end{proof} 
}

\ifshowproofs
 \myproofblockE
\fi

\begin{corollary}
The nucleolus can be found in polynomial time for the multi market uncapacitated production-distribution game with a fixed number of markets.
\end{corollary}

\subsection*{Remark}

The disjunctions we propose in this section can also be used to formulate the separation problem as a dynamic program and to exploit the framework of K\"onnemann and Toth~\citep{konemann2020general} (see Appendix~\ref{dpsep}). However, this approach yields less efficient algorithms.




\section{A combinatorial algorithms for the nucleolus of the uncapacitated single market game}\label{sec:comb}

In this section, we consider the single market case with  simplified notations.  Each company~$i$ has a demand $\demand_i$ and earns a per-unit profit $\alpha_i$, with profits ordered such that $\alpha_1 \geq \alpha_2 \geq \cdots \geq \alpha_n \geq 0$. The value of a coalition $S \subseteq N$ is given by $v(S) = \alpha_i \demand(S)$, where $i \in S$ is the company with the highest profit, i.e., the smallest index such that $S \cap {1, \dots, i-1} = \emptyset$. We rely on an instantiation of the Maschler scheme presented in Section~\ref{definition_nuc} where for $k\geq 1$, $\mathcal{F}^k=\{S: S\in span((\cup_{i\in F^k} \{\{i\}\}) \cup \{N\})\}$ for some $F^k\subseteq N\setminus 1$. 


In this case, we may restrict the family of sets $\mathcal{F}^k$ for constraints $x(S) = c^k(S)$ in (\ref{coreconsk}) to the family $\varphi^k:= (\cup_{i\in F^k} \{\{i\},  N\setminus i\}) \cup \{N\}$: it is easy to see that $\F^k=span(\varphi^k)$. Therefore, (\ref{coreconsk}) reduces to:

\vspace{-0.8cm}
\begin{align}
&\mbox{max } \epsilon \nonumber \\
(M^k) \ \ \ \ & x(S)\geq v(S) + \epsilon \ \ \  \forall S\in 2^N \setminus \mathcal{F}^k \label{coreconsk_} \\
&x(S) = c^k(S) \ \ \ \ \ \  \forall S\in \varphi^k. \nonumber
\end{align}

\vspace{-0.3cm}
$F^k$ represents a set of components for which the nucleolus value is known. We initially set $F^1=\emptyset$ (and therefore $\varphi^1=\{N\}$ and $\mathcal{F}^1=\{\emptyset,N\}$), and, at the end of each iteration $k$ we built  a strict superset $F^{k+1}$ of $F^{k}$ (still not containing $1$) that will be used for the next iteration until $F^{k+1}=N\setminus 1$. We also iteratively build an optimal solution $(x^k,\epsilon^k)$ to $(M^k)$ from an optimal solution $(x^{k-1},\epsilon^{k-1})$ to $(M^{k-1})$. Initially, for $k=1$, we use $(x^0,\epsilon^0)$, where $x^0_i=\demand_i \cdot \alpha_1$ for all $i\in N$ and $\epsilon^0=0$ ; which is feasible for $(M^1)$ by Lemma~\ref{dual}. This is done though Lemma~\ref{primal-dual}. The scheme will then converge in at most $n-1$ iterations. When $F^{k+1}= N\setminus 1$, $x^{k}$ is the nucleolus.  We point out that we will maintain the following property at each iteration $k\geq 1$:

\begin{center}
${\mathcal H}_k:$ for each $i \not\in F^k, i\neq 1$, $x^{k-1}(N\setminus i)=v(N\setminus i) + \epsilon^{k-1}$.    
\end{center}

\begin{lemma}\label{primal-dual}
 Let $(x^{k-1},\epsilon^{k-1})$ be a feasible solution to $(M^k)$ satisfying ${\mathcal H}_k$, for $k\geq 1$ and $F^k\subsetneq N\setminus 1$.  Let $\mu = \min_{S \not\in \F^k, 1\not\in S} \frac{x^{k-1}(S) - v(S) - \epsilon^{k-1} }{1+|S\setminus F^k|}$ and $\bar{S}$ a minimizer. The solution $(x^k,\epsilon^k)$ defined by

$$\begin{array}{l}
    {x}^k_i=\left\lbrace
\begin{array}{ll}
{x}^{k-1}_i +\mu \cdot (n-1-|F^k|)& \mbox{if } i=1,\\
{x}^{k-1}_i - \mu  & \mbox{if }i\not\in F^k, i\neq 1,\\
{x}^{k-1}_i & \mbox{otherwise.}
\end{array}
\right.\\
\\
\epsilon^k=\epsilon^{k-1}+\mu 

\end{array}$$

\noindent is an optimal solution to $(M^k)$. Moreover, the components of the nucleolus $\nu$ are such that $\nu_i=v(N)-v(N\setminus i)-\epsilon^k$ for $i\in \bar{S}\setminus F^k$ and, defining $F^{k+1}=\bar{S}\cup F^k$, $(x^k,\epsilon^k)$ is a feasible solution to $(M^{k+1})$ satisfying ${\mathcal H}_{k+1}$.

\end{lemma}
 
\newcommand{\myproofblockF}{ 
 \ifnotshowproofs
   \subsection{Proof of Lemma~\ref{primal-dual}}
 \fi 
\begin{proof}

First note that the problem ($M^k$) is the dual ($D^k)$ of problem $(P^k)$ below:\\

\begin{tabular}{cc}
\begin{minipage}[t]{0.55\textwidth}
\begin{tabular}{cc}
($P^k$) & $
\begin{array}{ll}
\min & - \sum_{S\in \varphi^k} c^k(S) y_S - \sum_{S\notin \F^k} v(S) y_S  \\
&A y = 0\\
&\sum_{S\not\in \F^k} y_S = 1\\
& y_S \geq 0, \forall S \not\in \F^k\\
\end{array}
$\\
\end{tabular}
\end{minipage}&
\begin{minipage}[t]{0.55\textwidth}
\begin{tabular}{cc}
($D^k$) & $
\begin{array}{ll}
\max & \epsilon\\
&x A_S \geq v(S) + \epsilon, \forall S \not\in \F^k \\
& x A_S = c^k(S), \forall S \in \varphi^k
\end{array}
$\\
\end{tabular}
\end{minipage}\\
\end{tabular}\\

where $A$ is a matrix whose columns $A_S$ are the characteristic vectors of the sets $S\in 2^N$. We will solve $(D^k)\equiv (P^k)$ using the primal-dual algorithm,  described in~\citep{papadimitriou1998combinatorial}, and applied to $(P^k)$. We follow the description and terminology from~\citep{papadimitriou1998combinatorial}. Interestingly the primal-dual algorithm converges in a single iteration in this setting! The primal-dual algorithm updates the starting feasible dual solution $(x^{k-1},\epsilon^{k-1})$ to $(D^k)$ through the iterative solution of a pair of auxiliary programs: the restricted primal ($RP$) and the dual of the restricted primal ($DRP$). 
 

\smallskip
 Suppose therefore that we are given a feasible solution $(x,\epsilon)$ to $(D^{k})=(M^k)$ with $F^k\subsetneq N\setminus 1$ that satisfies: for each $i \not\in F^k, i\neq 1$, $x(N\setminus i)=v(N\setminus i) + \epsilon$ : originally it is satisfied by $(x^{k-1},\epsilon^{k-1})$ by property $\mathcal{H}_k$.
As we want to apply the primal-dual algorithm to solve $(P^k)$,  we want to check optimality of $(x,\epsilon)$ using complementary slackness. It means we want to know whether there exists $y$ feasible solution to $P^k$ with $y_S=0$ whenever $S\not\in \F^k$ and $xA_S > v(S) + \epsilon$? The answer is ``yes" if and only if the optimal solution to the restricted primal $RP$ below has value 0 (we let ($DRP$) be the dual of ($RP$)):\\

\noindent\scalebox{0.9}{\begin{tabular}{cc}
\begin{minipage}[t]{0.5\textwidth}
\begin{tabular}[t]{cc}
($RP$) & $
\begin{array}{ll}
\min  & \delta \\
&A y = 0\\
&\sum_{S\not\in \F^k} y_S + \delta = 1\\
& y_S \geq 0, \forall S \not\in \F^k \\
& y_S = 0, \forall S\not\in \F^k: xA_S > v(S) + \epsilon\\
& \delta \geq 0\\
\end{array}
$\\
\end{tabular}
\end{minipage}&
\begin{minipage}[t]{0.5\textwidth}
\begin{tabular}[t]{cc}
($DRP$) & $
\begin{array}{ll}
\max & \tilde{\epsilon}\\
&\tilde{x} A_S  \geq \tilde\epsilon, \forall S \not\in \F^k: xA_S = v(S) + \epsilon  \\
& \tilde{x} A_S = 0, \forall S \in \varphi^k\\
& \tilde{\epsilon} \leq 1
\end{array}
$\\
\end{tabular}
\end{minipage}\\
\end{tabular}}\\

Observe first that ($RP$) has a (optimal) solution of value 0 if and only if there is a feasible solution to the following system:
\smallskip

\begin{tabular}{cc}
(I) & $
\begin{array}{l}
A y = 0\\
\sum_{S\not\in \F^k} y_S  > 0\\
 y_S \geq 0, \forall S \not\in \F^k \\
 y_S = 0, \forall S\not\in \F^k: xA_S > v(S) + \epsilon\\
\end{array}
$\\
\end{tabular}\\

as in this case we can scale $y$ to guarantee $\sum_{S\not\in \F^k} y_S =1$. Moreover, it follows that if ($RP$) does not have value $0$ then it has value $1$, as $y=0, \delta =1$ is a solution to ($RP$) (and if there is a solution ($RP$) with $1>\delta>0$, then we have a solution of value 0 to (I) and thus a solution to ($RP$), as we just observed).

\begin{claim}\label{drp}
Let $(x,\epsilon)$ be a feasible solution to $(D^k)$ with $F^k \subsetneq N\setminus 1$. Then:
\begin{itemize}
    \item[(i)] if there exists $\bar{S}$: $\bar{S}\not\in \F^k$, $x A_{\bar{S}} = v({\bar{S}}) + \epsilon$ and $1\not\in \bar{S}$, then $RP$ has value $0$;
    \item[(ii)] if vice versa each $\bar{S}$: $\bar{S}\not\in \F^k$, $x A_{\bar{S}} = v({\bar{S}}) + \epsilon$ is such that $1\in \bar{S}$, then $RP$ has value $1$
\end{itemize} 
\end{claim}
\begin{proof}

$(i)$. It is enough to prove that, in this case, there exists a solution $y$ that satisfies (I). Note that $A_{\bar{S}} + \sum_{i\in \bar{S}}  A_{N\setminus i} = |\bar{S}| A_N$. Hence, if we let $y_{\bar{S}} = 1$, $y_{N\setminus i} = 1$ for each $i\in \bar{S}$, $y_N = -|\bar{S}|$ and $y_S = 0$ for all the other sets $S$, we have $A y=0$. We now show that $y$ satisfies the other constraints in (I). First, $y$ satisfies constraints $y_S \geq 0, \forall S \not\in \F^k$, as the only negative value is $y_n$ and $N\in \varphi^k\subseteq \F^k$. Then $y$ satisfies constraints $y_S = 0, \forall S\not\in \F^k: xA_S > v(S) + \epsilon$. It is enough to check that all non zero values of $y$ are for sets $S$ with $S\in\F^k$ or $xA_S = v(S) + \epsilon$: $x A_{\bar{S}} = v({\bar{S}}) + \epsilon$ ;  $x(N\setminus i)=v(N\setminus i)+\epsilon$ for all $i\neq 1, i\not\in F^k$ by property $\mathcal{H}_k$ since $1\not\in \bar{S}$ ; $N\setminus i\in \varphi^k\subseteq \F^k$ for $i\in F^k$ ; and $N\in \varphi^k\subseteq \F^k$. Finally we have $\sum_{S\not\in \F^k} y_S >0$ by construction as the only negative value is for $y_N$ and $N\in \F^k$. Note that $y$ is then a dual optimal solution to $(D^k)$.

$(ii)$. Consider the solution to $DRP$ with $\tilde{\epsilon}=1$ and $\tilde{x}$ such that: $\tilde{x}(1)=n-1-|F^k|$; $\tilde{x}_i=-1$ for all $i \not\in F^k, i\neq 1$; $\tilde{x}_i=0$ for all $i\in F^k$. It has value $1$ so let us check that it is feasible. We check first that $\tilde{x}A_S=0$ for all $S\in \varphi^k$: for  $S=\{i\}$ for some $i\in F^k$,  we have $\tilde{x} A_S = 0$ ; for $S=N$  $\tilde{x} A_N= n-1-|F^k| - (|N \setminus F^k| -1)  = 0$ ;  for $S=N\setminus i$ for some $i\in F^k$, this is a consequence of the first two cases.  Now consider any set $S$ such that $S\not\in \F^k$ and $xA_S = v(S) + \epsilon$: by hypothesis $1\in S$. In this case, $\tilde{x} A_S = n-1-|F^k| - (|S\setminus F^k| -1)=|N\setminus F^k|-|S\setminus F^k|$. But $S\setminus F^k\neq N\setminus F^k$ as otherwise $S\in span(\cup_{i\in F^k} \{\{i\}\}\cup \{N\})$ and therefore $\tilde{x} A_S \geq 1=\tilde{\epsilon}$. Hence $(\tilde{x}, \tilde{\epsilon})$ is a feasible solution to $DRP$ with value 1. \hfill $\blacksquare$  
\end{proof}

{ Now we apply the previous claim to $(x^{k-1},\epsilon^{k-1})$}:  in order to check optimality of the current solution $(x^{k-1},\epsilon^{k-1})$, we only need to check whether there exists $\bar{S}$: $\bar{S}\not\in \F^k : x^{k-1} A_{\bar{S}} = v({\bar{S}}) + \epsilon^{k-1}$ with $1\not\in \bar{S}$. \\

\noindent (1) If there exists such a set $\bar{S}$, $(x^k,\epsilon^{k})$, which is equal then to $(x^{k-1},\epsilon^{k-1})$ as $\mu=0$, is an optimal solution to $(M^k)$ by Claim (i). Besides, from the proof of (i), we can exhibit a dual optimal solution for $(P^k)$ and by complementary slackness any optimal solution $(x,\epsilon)$ to $(D^k)$ will satisfy $x(N\setminus i)=v(N\setminus i)+\epsilon$ for all $i\in \bar{S}\setminus F^k$, and $x(N)=v(N)$ (by feasibility). This is the case in particular for $(x^{k},\epsilon^{k})$ and also for $(\nu,\epsilon^{k})$, where $\nu$ is the nucleolus. We can deduce $\nu_i=v(N)-v(N\setminus i) -\epsilon^{k}$ for all $i\in \bar{S}\setminus F^k$. We can then define $F^{k+1}=\bar{S}\cup F^k$: note there must exist an element in $F^{k+1}\setminus F^k=\bar{S}\setminus F^k$ otherwise $\bar{S}$ is a subset of $F^k$, and hence in $\F^k$, a contradiction. \\

\noindent (2) If vice versa, each $\bar{S}$: $\bar{S}\not\in \F^k$, $x^{k-1} A_{\bar{S}} = v({\bar{S}}) + \epsilon^{k-1}$ is such that $1\in \bar{S}$ then from the proof of fact (ii) in the claim, $(\tilde{x}, {\tilde{\epsilon}})$ -- with ${\tilde{\epsilon}}=1$, $\tilde{x}(1)=n-1-|F^k|$; $\tilde{x}_i=-1$ for all $i \not\in F^k, i\neq 1$; $\tilde{x}_i=0$ for all $i\in F^k$ -- is an optimal solution to $DRP$. According to the primal-dual algorithm, we can use $(\tilde{x},\tilde{\epsilon} )$ to update $(x^{k-1},\epsilon^{k-1})$. Namely, we look for the maximum  increase in the direction of $(\tilde{x},\tilde{\epsilon})$ by finding the maximum  $\tau\geq 0$ such that $(x^{k-1},\epsilon^{k-1}) + \tau (\tilde{x}, \tilde{\epsilon})$ remains feasible for $D$. Remember again from the proof of fact (ii) that: if $S\in \varphi^k$, then $\tilde{x} A_S = 0$; if $S\not\in \F^k$ and $1\in S$, then $\tilde{x} A_S\geq \tilde{\epsilon}$. Hence, the restriction can only come from a set $S\not\in \F^k$ with $1\not\in S$: in this case, $x^{k-1} A_S > v(S) + \epsilon^{k-1}$ (otherwise the value of $RP$ would be 0 - and we are in case (1)) and $\tilde{x} A_S < \tilde{\epsilon}$ (this holds for all $S$ not taking $1$). Thus $\tau$ can be defined as follows: 

\smallskip
$\tau = \min_{S \not\in \F^k, 1\not\in S} \frac{x^{k-1} A_S - v(S) - \epsilon^{k-1}}{\tilde{\epsilon} - \tilde{x} A_S}= \min_{S \not\in \F^k, 1\not\in S} \frac{x^{k-1}(S) - v(S) - \epsilon^{k-1}}{1+|S\setminus F^k|}=\mu$

\smallskip
 Note that in that case $\tau = \mu > 0$ as otherwise $RP$ would have value $0$. Let $\bar{S}$ be the argmin in the previous formula. We have $(x^k,\epsilon^k)=(x^{k-1},\epsilon^{k-1})+ \mu (\tilde{x},\tilde{\epsilon})$.  Obviously $(x^k,\epsilon^k)$ satisfies $x^k A_{\bar{S}} = v({\bar{S}}) + \epsilon^k$ and,  by the claim applied to $(x^k,\epsilon^k)$, $RP$ has now value zero and is thus optimal. This proves that the solution $(x^{k},\epsilon^{k})$ is optimal for $(M^k)$ and we can conclude that $\nu_i=v(N)-v(N\setminus i) -\epsilon^{k}$ for all $i\in \bar{S}\setminus F^k$ and define $F^{k+1}=\bar{S}\cup F^k$ following the arguments in (1).

We are left to prove that $(x^k,\epsilon^k)$ satisfies ${\mathcal H}_{k+1}$, meaning that for each $i \not\in F^{k+1}, i\neq 1$, $x^{k}(N\setminus i)=v(N\setminus i) + \epsilon^{k}$. By hypothesis, we have $x^{k-1} A_{N\setminus i}=v({N\setminus i}) + \epsilon^{k-1}$ for all $i\neq 1$, $i\not\in F^k$. If $\mu=0$, $x^{k}(N\setminus i)=v(N\setminus i) + \epsilon^{k}$ holds for each $i \not\in F^{k+1}, i\neq 1$ as $(x^k,\epsilon^k)=(x^{k-1},\epsilon^{k-1})$ and $F^{k+1}\supset F^k$. Else, by construction, we have that $\tilde{x} A_{N\setminus i}= \tilde{\epsilon}$ for all $i\neq 1$, $i\not\in F^k$ and thus also $x^k A_{N\setminus i}=v({N\setminus i}) + \epsilon^k$ for all $i\neq 1$,  $i\not\in F^k$, and therefore for $i\neq 1, i\not\in F^{k^+1}$ as well since $F^{k+1}\supset F^k$.

\qed  \end{proof} 
}
\ifshowproofs 
 \myproofblockF
\fi 
Lemma~\ref{primal-dual} can be used iteratively until $F^{k+1}=N\setminus 1$ starting from $(x^0,\epsilon^0)$ where $x^0_i=\demand_i \cdot \alpha_1$ for all $i\in N$ and $\epsilon^0=0$ (it is feasible for $(M^1)$ by Lemma~\ref{dual} as already discussed and it satisfies trivially property $\mathcal{H}_1$. We therefore have the following theorem.

 \begin{theorem}\label{th:main}
 The nucleolus can be found in  time $O(n^4)$ for single market uncapacitated production-distribution games.
 \end{theorem}
 
\newcommand{\myproofblockG}{ 
 \ifnotshowproofs
    \subsection{Proof of Theorem~\ref{th:main}}
 \fi 
\begin{proof}
We simply need to show that the point $(x^k,\epsilon^k)$ can be computed in time $O(n^3)$ in Lemma~\ref{primal-dual}.  

 The computation of $\mu$ is easy. We enumerate over all choices for $i\neq 1$ and over the possible values for $t:=|S\setminus F^k|$, from $1$ to at most $n-1-|F^k|$. For a fixed $i$ and $t$, we solve a problem of the form $\min_{S \in \mathcal{S}_i: S\notin \F^k, |S\setminus F^k|=t} \frac{w(S) - \epsilon^{k-1}}{1+t}$ (hence we have to solve $O(n^2)$ problems), where $\mathcal{S}_i$ is the family of sets $S$ such that $i\in S$ and $S\cap \{1, 2, \ldots, i-1\}=\emptyset$, and $w_l = x_l - \alpha_i \cdot \demand_l$ for $l=i,...,n$. This problem is equivalent to $\min_{S \in \mathcal{S}_i: S\notin \F^k, |S\setminus F^k|=t} w(S)$ (as $t$ and $\epsilon^{k-1}$ are fixed) and then to $\min_{S \in \mathcal{S}_i: S\not\subseteq F^k, |S\setminus F^k|=t} w(S)$. To see the latter equivalence, observe first that $S\in \F^k=span(\cup_{i\in F^k} \{\{i\}\} \cup \{N\})$ if and only if $S\subseteq F^k$ or $S=N\setminus T$ with $T\subseteq F^k$. Hence $S\not\in \F ^k$ means simultaneously $S\not \subseteq F^k$ and $N\setminus S\not\subseteq F^k$. Now, since $i\neq 1$, we have that for all $S\in \mathcal{S}_i$, $1\notin S$, so $N\setminus S\not\subseteq F^k$ is always satisfied. The solution is then trivial: we take $i$ and all elements $j$ that are in $F^k\cap\{i+1, \ldots, n\}$ such that $w_j<0$ and we add either the smallest (with respect to the weight $w$) $t$ or $t-1$ elements (depending whether $i\in F^k$ or not) in $(N\setminus F^k)\cap\{i+1, \ldots, n\}$.

\qed
\end{proof} 
}

\ifshowproofs 
 \myproofblockG
\fi 




\citet{Granot98} and~\citet{reijnierse98} independently introduced the concept of {\em characterization set}, which is a collection of coalitions ${\cal Q}\subseteq 2^N$ that determine the nucleolus by itself. In other words, ${\cal Q}$ is a characterization set if replacing $e(x) = (e_S(x))_{S\in 2^N}$ by $e^Q(x) = (e_S(x))_{S\in {\cal Q}}$ in the computation of the nucleolus does not affect the nucleolus itself. 
\citet{huberman80} and~\citet{Granot98} were able to provide sufficient conditions for a set of coalitions to be a characterization set; a good overview of these results is given in~\citet{skizlai}. To the best of our knowledge, none of these conditions is useful to show that in our framework there exists a characterization set of polynomial length; in contrast, when $\alpha_1 > ... > \alpha_n > 0$ and $d_i >0, i\in N$, these conditions — once again they are just sufficient — imply that ${\cal Q} = 2^N$. However, we can claim the following:

\begin{theorem}\label{characterization}
For the nucleolus of the single market uncapacitated production-distribution game, there always exists a family of $2n-1$ sets that is a characterization set. 
 \end{theorem}

\newcommand{\myproofblockH}{ 
 \ifnotshowproofs
    \subsection{Proof of Theorem~\ref{characterization}}
 \fi 
\begin{proof}
It is enough to choose the characterization set ${\cal Q}$ as the union of the $n-1$ sets $N\setminus \{i\}, i\neq 1$, the set $N$ and, for each iteration of the primal-dual algorithm, a set $\bar S$ that is a minimizer for the value of $\mu$ (see the statement and the proof of Lemma~\ref{primal-dual}).

\qed
\end{proof} 
}

\ifshowproofs 
 \myproofblockH
\fi

However, we do not know how to find this family, or any other of polynomial length, {\em a-priori}.

\section{Open questions and a final remark}\label{conclusion}

While we have shown that production-distribution games are linear production games, it remains unclear whether they coincide with this class. 
The most intriguing open question is however whether computing the nucleolus is computationally hard for the {\em capacitated} production-distribution game already for the {\em single} market case. Another open question is whether computing the nucleolus is computationally hard in the uncapacitated setting when the number of markets is unbounded.

Concerning the latter question, recall that by Remark~\ref{decomposition}, in the uncapacitated setting, it is possible to associate to the multi market game $(N,v)$ a single market game $(N, v_j)$ for each market $j\in M$, so that $(N, v)=(N,\sum_{j\in M} v_j)$. We believe that, although the solution provided by the sum of the nucleolus of the single market games $(N, v_j)$ does not in general coincide with the nucleolus of $(N,v)$ (see Example~\ref{theexample}), it is still as a sensible solution in practice. In fact, the solution to each single market game $(N, v_j)$ can be computed quickly using the combinatorial algorithm presented in Section~\ref{sec:comb} and it seems reasonable to argue on each market separately. Besides, the corresponding solution is in the core of $(N,v)$ as, the core of each game $(N,v_j)$ is nonempty, the nucleolus is in this core (as always, for games with nonempty core), and the sum of solutions in the core of $(N,v_j)$ is obviously in the core of $(N,\sum_{j\in M} v_j)=(N,v)$.

A last open question concerns Theorem~\ref{characterization}. Is it is possible to find a-priori a characterization set of polynomial length? Some results in this direction are shown in \citet{jens24} for the so-called happy nucleolus of set covering games.

\bibliographystyle{plainnat}
\bibliography{biblio}

\appendix

\newpage

\newpage

\ifnotshowproofs 
\section{Proofs}

\myproofblockA
\myproofblockB
\myproofblockC
\myproofblockD
\myproofblockE
\myproofblockF
\myproofblockG
\myproofblockH

\fi 

\section{Computing the nucleolus with the framework of K{\"o}nemann and Toth}\label{dpsep}

We show how to exploit the disjunctions presented in Section~\ref{singleuncap} to leverage the framework introduced by K{\"o}nemann and Toth~\citep{konemann2020general} and compute alternatively the nucleolus of uncapacitated production-distribution games with a fixed number of markets in polynomial time. 

K{\"o}nemann and Toth~\citep{konemann2020general} showed that the nucleolus can be computed via the Maschler scheme~\citep{maschler1979geometric} in polynomial time when the {\em minimum excess coalition problem} can be solved in polynomial time via an {\em integral dynamic program}. The minimum excess coalition problem is the following: given $x\in\setR^n$, find $S\in \argmin \{x(S)-v(S), S\subseteq N\}$. A formal definition of integral dynamic programs is given in~\citep{konemann2020general}. For this paper, it is enough to know that dynamic programs that can be phrased as shortest path problem in directed acyclic fit this framework~\citep{konemann2020general}.

\subsection{Single market}

 We show first that, for the uncapacitated production-distribution game with a single market, the minimum excess coalition problem can be indeed solved in polynomial time through shortest path in a directed acyclic graph. In the single market game, we simplify notations and assume the following. The commodity product is sold at market price $r$ and the (per unit) production cost for company $i$ is $c_i$. Each company $i$ owns a demand $d_i$. We let $\alpha_i:= r-c_i$ and we assume here without loss of generality that $\alpha_1\geq \alpha_2 \geq ... \geq \alpha_n \geq 0$.  the value $v(S)$ of a coalition $S\subseteq N$ in the single market game is therefore equal to:

\begin{equation}\label{sinvuncap}
v(S) = (\sum_{i\in S} d_i) \cdot \max_{k\in S}\alpha_k.
\end{equation}

  We enumerate over the $n$ possible values in $\{1, \ldots, n\}$ for the minimum index $j$ of a player in a feasible solution $S\neq \emptyset$  to the minimum excess coalition problem: then the problem $\argmin \{x(S)-v(S), S\subseteq N \}$ decomposes into $n$ problems of the form $\argmin \{x(S)-\alpha_j \cdot \demand(S): S\subseteq \{j,...,n\} \mbox{ and } j\in S\}$ (plus a comparison with $x(\emptyset)-v(\emptyset)=0$). The latter problem is of the form $\argmin \{\sum_{i\in S} y_i : S\subseteq \{j,...,n\} \mbox{ and } j\in S\}$, with $y_i:=x_i - \alpha_j \cdot \demand_i$. The problem can be modeled as a shortest path problem -- of course, it is using a sledgehammer to crack a nut! -- from node $j$ to node $n+1$ in a directed acyclic graph $G_j$. The graph $G_j$ has $n-j+2$ nodes (one node for each element in $\{j,...,n\}$, and one extra node, numbered $n+1$, representing a sink), there is an arc $(p,q)$ if and only if $p, q\in \{j, \ldots, n+1\}$ and $p<q$, and the cost of an arc $(p,q)$ is simply $y_p$. The nodes of a shortest path from $j$ to $n+1$, different from $n+1$, are  the elements of a coalition $S$ with minimum excess among the coalitions $S\subseteq \{j,...,n\}$ with $j\in S$. 

 Now observe that the disjunction over the different values of $j$ can be represented by a shortest path from a source node $s$ to a sink node $t$ in a larger directed acyclic graph obtained by taking the disjoint union of the directed acyclic graphs $G_j$, $j=1,...,n$ and adding a node $s$ with an arc of zero cost to the source node $j$ of $G_j$ for any $j=1,...,n$ and a node $t$ with an arc of zero cost from each sink node $n+1$ of $G_j$. In addition, we can add an arc from $s$ to $t$ of cost zero to model the empty set. 

\subsection{Fixed number of market}

For the multi market case, we can follow a similar rationale. Remember that the value $v(S)$ of a coalition $S\subseteq N$ is equal to (see~\eqref{vuncap}):

$$v(S) = \sum_{j\in M} (\sum_{i\in S} d_{ij}) \cdot \max_{k\in S}\alpha_{kj}.$$

We enumerate over the (at most) $n^{m}$ possible values of the minimum index vector $(\argmax_{i\in S} \alpha_{ij} ,j=1,...,m)$ of a feasible solution  $S\neq \emptyset$ to the minimum excess problem. Let $\mathcal{S}^k_j:=\{i\in N:\alpha_{ij} \leq \alpha_{kj}\}$. The problem $\argmin \{x(S)-v(S), S\subseteq N \}$ then decomposes into (at most) $n^m$ problems of the form $$\argmin \{x(S)- \sum_{j\in M} \alpha_{i_j,j} \cdot \sum_{i\in S} d_{ij}: S  \ s.t. \{{i}_1, ...,{i}_m\}\subseteq S\subseteq \cap_{j\in M} \mathcal{S}^{i_j}_j\}$$
for any choice of $({i}_1, ...,{i}_m) \in [n]^m$ (plus a comparison with $x(\emptyset)-v(\emptyset)=0$). Note that some of the problems might be unfeasible (when there exists $j$ such that $i_j\not \in \cap_{j\in M} \mathcal{S}^{i_j}_j\}$ ) and we simply discard those. The latter problem is of the form $\argmin \{\sum_{i\in S} y_i: S  \ s.t. \{{i}_1, ...,{i}_m\}\subseteq S\subseteq \cap_{j\in M} \mathcal{S}^{i_j}_j\}$ with $y_i=x_i - \sum_{j\in M} \alpha_{i_j,j} \ d_{ij}$. Each problem can again be formulated as a shortest path problem in a directed acyclic graph $G_{i_1,...,i_m}$. The graph $G_{i_1,...,i_m}$ has node set $V=\{0\}\cup \cap_{j\in M} \mathcal{S}^{i_j}_j \cup\{n+1\}$ and there is an arc $(p,q)$ if and only if $p, q\in V$ and $p<q$. The cost of an arc $(p,q)$ is $y_p$ if $p\neq 0$ and $\sum_{j\in M} y_{i_j}$ if $p=0$. The disjunction over the different (feasible) values of $({i}_1, ...,{i}_m)$) and $S=\emptyset$ can be treated similarly to the single market case. When $m$ is fixed the corresponding graph is polynomial in the input size and so is the shortest path problem and thus the computation of the nucleolus thanks to~\citep{konemann2020general}.

\begin{remark}
   Note that the corresponding algorithms are computationally prohibitive, and the other approaches proposed  in this manuscript are considerably more efficient. 
\end{remark}

\end{document}